\documentclass[letterpaper]{article} 
\usepackage{aaai24}  
\usepackage{times}  
\usepackage{helvet}  
\usepackage{courier}  
\usepackage[hyphens]{url}  
\usepackage{graphicx} 
\usepackage{natbib}  
\usepackage{caption} 
\usepackage{algorithm}
\usepackage{algorithmic}

\usepackage{graphicx}
\usepackage{enumitem}
\usepackage{amsmath}
\usepackage{amsthm}
\usepackage{amssymb}
\usepackage{xcolor}
\usepackage{booktabs}
\usepackage{pifont}
\newcommand{\cmark}{\ding{51}}%
\newcommand{\xmark}{\ding{55}}%

\newcommand{\rom}[1]{\uppercase\expandafter{\romannumeral #1\relax}}

\newcommand{\eat}[1]{}

\newcommand{\A}{\mathcal{A}}
\newcommand{\B}{\mathcal{B}}
\newcommand{\N}{\mathcal{N}}
\newcommand{\M}{\mathcal{M}}
\newcommand{\Z}{\mathbb{Z}}

\usepackage{lineno}

\newtheorem{theorem}{Theorem}
\newtheorem{lemma}[theorem]{Lemma}
\newtheorem{proposition}[theorem]{Proposition}

\newtheorem{definition}[theorem]{Definition}

\usepackage{newfloat}
\usepackage{listings}
\DeclareCaptionStyle{ruled}{labelfont=normalfont,labelsep=colon,strut=off} 
\floatstyle{ruled}
\newfloat{listing}{tb}{lst}{}
\floatname{listing}{Listing}
\title{Almost Envy-Free Allocations of Indivisible Goods or Chores with Entitlements}
\author{
    MohammadTaghi Hajiaghayi\textsuperscript{\rm 1}, Max Springer \textsuperscript{\rm 2}, Hadi Yami\textsuperscript{\rm 3}
}
\affiliations{
    \textsuperscript{\rm 1}University of Maryland, Department of Computer Science, College Park MD\\
    \textsuperscript{\rm 2}University of Maryland, Department of Mathematics, College Park MD\\
    \textsuperscript{\rm 3}Microsoft Corporation, Redmond WA
    
    \{hajiagha, mss423\}@umd.edu, hayami@microsoft.com
}

\usepackage{bibentry}

\begin{document}

\maketitle

\begin{abstract}
We here address the problem of fairly allocating indivisible goods or chores to $n$ agents with \emph{weights} that define their entitlement to the set of indivisible resources. Stemming from well-studied fairness concepts such as envy-freeness up to one good (EF1) and envy-freeness up to any good (EFX) for agents with \textit{equal} entitlements, we present, in this study, the first set of impossibility results alongside algorithmic guarantees for fairness among agents with \textit{unequal} entitlements.

Within this paper, we expand the concept of envy-freeness up to any good or chore to the weighted context (WEFX and XWEF respectively), demonstrating that these allocations are not guaranteed to exist for two or three agents. Despite these negative results, we develop a WEFX procedure for two agents with \emph{integer} weights, and furthermore, we devise an approximate WEFX procedure for two agents with \emph{normalized} weights. We further present a polynomial-time algorithm that guarantees a weighted envy-free allocation up to one chore (1WEF) for any number of agents with additive cost functions. Our work underscores the heightened complexity of the weighted fair division problem when compared to its unweighted counterpart.
\end{abstract}

\section{Introduction}
The notion of a fair allocation of resources has been a fundamental problem in the field of economics and, more recently, computer science \cite{Bouveret_2016,Brams_1995,Peterson_2002,Peterson_2009,Pikhurko_2000,Procaccia_2009,Robertson_1998}. This idea of fairness takes on many forms in a wide array of application areas such as land division, divorce settlements, or natural resource distributions. 

One of the most well-established mathematical definitions of fairness is \emph{envy-freeness}. In an envy-free allocation each agent prefers her share more than any other: for the case of goods, she wants the most lucrative set of items whereas for chores, we seek to minimize the cost distributed across agents. However, envy-freeness is not necessarily obtainable when dealing with a set of indivisible goods or chores, and such an allocation is often difficult to compute. Hence, several relaxations of envy-freeness have been introduced.

%

Two such relaxations that we examine in the present work are \emph{envy-free up to one good} (EF1) and \emph{envy-free up to any good} (EFX). An allocation of indivisible items is EF1 if any possible envy of an agent for the share of another agent can be resolved by removing some good from the envied share. \cite{Lipton_2004} and \cite{Budish_2010} provide polynomial time algorithms for an EF1 allocation for any number of agents. An allocation is said to be EFX if no agent envies another agent after the removal of \emph{any} item from the other agent’s bundle. Theoretically, this notion is strictly stronger than EF1 and as a result, despite significant effort, the existence of EFX allocations is still unknown. \cite{Plaut_2018} described the first \emph{approximate} EFX results with later improvements by \cite{Amanatidis_2020_2} and \cite{Farhadi_2021}. With respect to chore division, the literature is more scarce: the first discrete and bounded envy-free protocol for any number of agents was not proposed until 2018 \cite{dehghani2018envy}.

In both the good and chore division problems, the literature noted above is restricted to the case where each agent's perspective is weighted equally. However, more recently, the generalization has been proposed where agents have a valuation (or cost in the case of chores) associated with each item as well as an ``entitlement" or weighting \cite{Aziz_2019_2,Babaioff_2021,Chakraborty_2020,Farhadi_2017}. This problem setting captures significantly more characteristics of real-world fair division issues than that of the equal entitlement case. 

To further emphasize this natural problem, consider the recent example that arose during the height of the COVID-19 pandemic. When a vaccine was in production, producers were faced with the question of how to fairly distribute doses in a manner that respects all parties while also mitigating further spread of the virus. In this context, it is not enough to just give each nation an equal share of the supply -- we must now take into account population densities and each nation's medical infrastructure among a plethora of other factors. This naturally gives rise to agent weightings that, as we will show in the paper, result in a considerable complication of the fair allocation problem.


\begin{table*}[t]
\centering
\begin{tabular}{c|cccccc}
 &
  \textbf{1WEF} &
  \textbf{WEFX (id.)} &
  \textbf{WEFX (int.)} &
  \textbf{WEFX (add.)} &
  \textbf{XWEF} &
  \textbf{$\alpha$-WEFX} \\ \hline
\textbf{$n=2$} &
  \cmark (Thm. \ref{thm:1wef}) &
  \cmark (Thm. \ref{thm:ord_wefx}) &
  \cmark (Thm. \ref{thm:int_wefx}) &
  \xmark (Thm. \ref{thm:wefx_n2}) &
  \xmark (Thm. \ref{thm:xwef_n2}) &
  \cmark (Thm. \ref{thm:alg}) \\ \hline
\textbf{$n=3$} &
  \cmark (Thm. \ref{thm:1wef}) &
  \cmark (Thm. \ref{thm:ord_wefx}) &
  $\mathord{?}$ &
  \xmark (Thm. \ref{thm:wefx_n3}) &
  \xmark (Thm. \ref{thm:xwef_n3}) &
  $\mathord{?}$ \\ \hline
\textbf{$n>3$} &
  \cmark (Thm. \ref{thm:1wef}) &
  \cmark (Thm. \ref{thm:ord_wefx}) &
  $\mathord{?}$ &
  $\mathord{?}$ &
  $\mathord{?}$ &
  $\mathord{?}$
\end{tabular}
\caption{A summary of our existence results. ``\cmark'' indicates the type of allocation specified by the column is guaranteed to exist in the setting specified by the row, while ``\xmark'' indicates that we give a counterexample and ``$\mathord{?}$'' indicates an open question. ``id.'', ``int.'', and ``add.'' serve as shorthand for identical valuation, integer valued weights, and additive valuation assumptions respectively.}
\label{tab:1}
\end{table*}

\subsection{Our Contributions}
Although the existence of EFX allocations is a major open problem in the field, we demonstrate in Section~\ref{sec:imposs} that WEFX, the generalized version of EFX for agents with entitlements, cannot be guaranteed in general with counterexamples that are at best 0.786-WEFX for $n=2$ agents (or 0.795 for $n=3$), giving a novel upper bound on the problem. We further provide analogous impossibility results for the chore setting on $n=2$ or 3 agents in Section~\ref{sec:xwef}. 

Nevertheless, we showcase that, for the case of two agents, a variant of the ``I-cut-you-choose" procedure ensures the WEFX property when dealing with integer-valued weights. Additionally, we introduce a novel procedure for an approximate WEFX solution under normalized weights in Section~\ref{sec:alg}, employing an algorithm akin to the extensively studied ``moving-knife" technique. This approach attains an approximate factor of $\frac{w}{2\sqrt[3]{m}}$, where $w$ represents the highest agent weighting and $m$ is the number of items. Synthesis of the above results demonstrates that the weighted version of the envy-free up to any good (chore) problem is \emph{considerably} more challenging than its unweighted counterpart, even in the most simplistic case of $n=2$. These results also nicely address an open direction posed in \cite{Chakraborty_2020} on this notion of weighted EFX allocations.

Finally, in Section~\ref{sec:1wef}, we extend the work of \cite{Chakraborty_2020} to give a weighted-picking sequence procedure for allocating indivisible chores to weighted agents that is \textit{envy-free up to one chore (1WEF)} and runs in polynomial time. This positive result further reaffirms the gap in complexity between our two relaxed notions of envy-freeness. We here note that independent and in parallel to our work, \cite{wu2023weighted} also proved such allocations exist and can be computed efficiently using a similar procedure. A summary of our existence and impossibility results is presented in Table~\ref{tab:1}.

\subsection{Further Related Work}
An extensive line of work exists concerning fair division of indivisible items, and for a complete overview we defer the reader to the surveys of \cite{chevaleyre2017distributed} and \cite{markakis2017approximation}. We here focus only a subset of papers that are most relevant to the present work.

Prior works pertaining specifically to the fair allocation of indivisible items to \emph{asymmetric} agents have been predominantly focused on fairness notions not based on envy. \cite{Farhadi_2017} introduced the weighted extension of the \emph{maximin share} (MMS) as studied by \cite{Barman_2020,Budish_2010,Feige_2021,Garg_2019}. Moreover, \cite{Aziz_2019_2} explored the weighted MMS (WMMS) concept for fair division of chores (negatively valued goods) where an agent's weight is intuitively their share of the overall workload to be completed. \cite{babaioff2019fair} examine the competitive equilibrium of agents with asymmetric budgets. More recently, they worked to redefine and further investigate the MMS allocation for the case of agents with arbitrary entitlements \cite{Babaioff_2021}. The authors first note that a WMMS allocation, as presented in \cite{Farhadi_2017}, does not align with the intuitive sense of ensuring that highly entitled agents receive a stronger preferential treatment in the allocation process. As such, they define the \textit{AnyPrice Share} (APS) allocation for both the unweighted and weighted contexts and provide a $\frac{3}{5}$ approximate solution.

The recent work of \cite{aziz2020polynomial} exhibited a polynomial-time algorithm for the computation of an allocation that satisfies both Pareto optimality and \emph{weighted proportionality up to one item} (WPROP1) of goods and chores for agents with asymmetric weights. A large body of work examines unequal agent weightings for divisible items through the notion of \emph{proportionality} \cite{barbanel1996game,Brams_1995,cseh2020complexity,segal2020cut}. While in the unweighted setting, EF1 allocations are also PROP1, the weighted setting does not admit such a luxury \cite{Chakraborty_2020}. 

Most closely related to our paper is the recent work by Chakraborty et al. \cite{Chakraborty_2020} which considers the allocation problem for indivisible goods with the generalized notion of weighted agents, providing a polynomial time WEF1 alogirthm. We build upon this result and the algorithm for the case of chores, while also providing the first results for both WEFX and XWEF with various existential results for each -- answering the question of whether these generalized notions of fair allocations can be computed. 

\section{Preliminaries and Basic Definitions}
\paragraph{Fair Allocation Problem.} An instance of a \emph{weighted} fair allocation problem consists of a set $\N$ of $n$ agents where each $i \in \N$ has weight $w_i > 0$. Let $\M$ be a set of $m$ goods (or chores) with, potentially heterogenous, valuation functions $v_i : 2^m \rightarrow \mathbb{R}_{\ge 0}$ (or cost functions $c_i : 2^m \rightarrow \mathbb{R}_{\ge 0}$) that are assumed to be additive unless otherwise noted. An allocation of $\M$ is a partition of the set into $n$ disjoint ``bundles'', $\A = (\A_1, ..., \A_n)$, such that $\bigcup_{i \in [n]} \A_i = \M$ and $\A_i \cap \A_j = \emptyset$ for any two $i,j \in [n]$. For readability, we will henceforth let $\A = (\A_1, ..., \A_n)$ denote an allocation of goods and $\B = (\B_1,...,\B_n)$ chores. 


\paragraph{Fairness Criteria.} We seek to find an allocation over the set of goods (or chores) that is \emph{envy-free}: given an instance of a fair division problem and an allocation $\A$, an agent $i$ envies agent $j$ if they strictly prefer the set $\A_j$ over their own bundle $\A_i$ up to a scaling by their weight.
\begin{definition}
An allocation is \textit{weighted envy-free} (WEF) if no agent envies another, i.e. for any pair $i,j \in \N$ of agents we have
\begin{align*}
	\frac{v_i(\A_i)}{w_i} \geq \frac{v_i(\A_j)}{w_j} \text{ for goods, } \frac{c_i(\B_i)}{w_i} \leq \frac{c_i(\B_j)}{w_j} \text{ for chores.}
\end{align*}
\end{definition}
However, envy-freeness is too strong a criteria to meet for both indivisible goods and chores. As such, we define two relaxations of this notion, namely \textit{weighted envy-freeness up to one good or chore} (WEF1 and 1WEF respectively) and \textit{weighted envy-freeness up to any good or chore} (WEFX and XWEF respectively).
\begin{definition}
An allocation of goods $\A$ (or chores $\B$) is said to be (\textit{i}) weighted envy-free up to one good (WEF1) if for any pair of agents $i,j \in \N$ if
\begin{align*}
	\frac{v_i(\A_i)}{w_i} \geq \text{min}_{a \in \A_j}\frac{v_i(\A_j \setminus \{a\})}{w_j}
\end{align*}
(\textit{ii}) weighted envy-free up to one chore (1WEF) if for any pair of agents $i,j \in \N$ if
\begin{align*}
	 \text{min}_{b \in \B_i}\frac{c_i(\B_i \setminus \{b\})}{w_i} \leq \frac{c_i(\B_j)}{w_j}
\end{align*}
\end{definition}
\begin{definition}
An allocation of goods $\A$ (or chores $\B$) is said to be (\textit{i}) weighted envy-free up to any good (WEFX) if for any pair of agents $i,j \in \N$ if
\begin{align*}
	\frac{v_i(\A_i)}{w_i} \geq \max_{a \in \A_j}\frac{v_i(\A_j \setminus \{a\})}{w_j}
\end{align*}
(\textit{ii}) weighted envy-free up to any chore (XWEF) if for any pair of agents $i,j \in \N$ if
\begin{align*}
	 \max_{b \in \B_i}\frac{c_i(\B_i \setminus \{b\})}{w_i} \leq \frac{c_i(\B_j)}{w_j}
\end{align*}
\end{definition}
While the envy-freeness up to one and any item are very closely related, they define relaxed notions of fairness that have a large discrepancy in terms of complexity. The current literature has demonstrated that a WEF1 allocation always exists for agents with additive valuation functions \cite{Chakraborty_2020}, a result that we here extend to the chore division setting with a polynomial time algorithm. However, the problems of the existence of WEFX allocations for goods as well as XWEF for chores remain open.
In this paper, we show that WEFX allocations are guaranteed to exist under certain, restrictive, problem assumptions. However, in general, we show that WEFX and XWEF allocations are not guaranteed to exist for agents with additive cost functions, but overcome this inherent barrier with a polynomial time algorithm that yields an \emph{approximate} WEFX allocation for two agents.
\begin{definition}
For constants $\alpha \leq 1$ and $\beta \geq 1$, an allocation of goods $\A$ (or chores $\B$) is called (\textit{i}) $\alpha$-approximate envy-free up to any good ($\alpha$-WEFX) if for any $i,j \in \N$
\begin{align*}
	\frac{v_i(\A_i)}{w_i} \geq \alpha \cdot \max_{a \in \A_j}\frac{v_i(\A_j \setminus \{a\})}{w_j}
\end{align*}
(\textit{ii}) $\beta$-approximate envy-free up to any chore ($\beta$-XWEF) if for any $i,j \in \N$
\begin{align*}
	\max_{b \in \B_i}\frac{c_i(\B_i \setminus \{b\})}{w_i} \leq \beta \cdot \frac{c_i(\B_j)}{w_j}
\end{align*}
\end{definition}


\section{Weighted Division of Indivisible Goods}
We here investigate the issue of fairly distributing a collection of indivisible goods among weighted agents, considering different problem assumptions. While achieving envy-freeness through a complete allocation is the ideal outcome, it is not always easily computed or assured by an algorithm. Consequently, we often turn to the less strict criteria of envy-freeness up to any item. The overall existence of such criteria remains an unresolved challenge in the realm of fair allocation for both unweighted and weighted systems.

In this context, we present certain assumptions that prove adequate for ensuring the existence of WEFX allocations. Moreover, when these assumptions are relaxed, we exhibit the non-existence of such allocations.

\subsection{Identical Valuations}
To begin, we simplify the scenario by initially considering a situation where the valuation functions are the same across all agents. In this setting, we observe that the classic \emph{envy-cycle elimination} method by Lipton produces the desired WEFX allocation with only a minor adjustment. Specifically, the ``envy-graph'' is modified to have an edge from agent $i$ to agent $j$ if and only if $i$ has \emph{weighted} envy towards $j$. This gives us the following positive results.

\begin{theorem} \label{thm:add_wefx}
    The weighted version of \cite{Lipton_2004}'s envy-cycle elimination algorithm produces a complete WEFX allocation for agents with (additive) identical valuation functions.
\end{theorem}
\eat{
\begin{proof}
    Let $v$ denote the valuation function for all agents. By construction of the algorithm \cite{Lipton_2004}, the (incomplete) allocation at the end of each iteration is guaranteed to be WEFX as long as we can find an agent, say $i$, towards whom no other agent has weighted envy at the beginning of the iteration. This agent will pick their favorite among the remaining items which will then be the least valuable item in $i$'s bundle according to all agents due to the identical valuations (and the greedy selection by $i$). At this step in the allocation procedure, we give the item under consideration to agent $i$ and thus any resulting weighted envy towards $i$ can be eliminated by removing this item. If there is no unenvied agent, then the weighted envy graph consists of at least one cycle; however, under identical valuations, the envy graph cannot have cycles. To see this, suppose that agents $1,2,...,\ell$ form a cycle (in that order) for some $\ell \in [n]$. Since agents have identical valuations, it must be that
    
    \begin{align*}
        \frac{v(\A_1)}{w_1} < \frac{v(\A_2)}{w_2} < ... < \frac{v(\A_\ell)}{w_\ell} < \frac{v(\A_1)}{w_1},
    \end{align*}
    which yields a contradiction.
\end{proof}
}

Moreover, we can slightly generalize this result to encompass the setting in which agents share a preferential ordering over the set of goods.
\begin{theorem} \label{thm:ord_wefx}
    The weighted version of the envy-cycle elimination algorithm produces a complete WEFX allocation for agents with identical \underline{ordinal} preferences.
\end{theorem}
The proofs of both of these theorems are deferred to the appendix due to space constraints.

The simple model of identical valuations however affords results that are considerably more positive than the general case as exhibited by the below result.
\begin{proposition}[\cite{Chakraborty_2020}] \label{prop:chakra}
    The weighted version of the envy-cycle elimination algorithm may not produce a complete WEF1 allocation, even in a problem instance with two agents and additive valuations.
\end{proposition}

\subsection{General Impossibility} \label{sec:imposs}
For $n=2$ and 3, it is known that EFX allocations are guaranteed to exist under additive valuation functions \cite{Plaut_2018,chaudhury2020efx}. However, we here present novel existential results that show the same result does \emph{not} hold in the weighted setting -- a considerable deviation between the two problem contexts. 

We first sketch the construction for a constant factor approximation upper bound to the WEFX problem on two agents, presenting a complete analysis of the result in the appendix.

\begin{theorem}\label{thm:wefx_n2}
For $n=2$ agents with additive valuation functions and weights $w_1$ and $w_2$ such that $w_1 + w_2 = 1$, there exists an instance where a complete WEFX allocation of goods does not exist while a 0.786-WEFX allocation does.
\end{theorem}
\begin{proof}[Proof Sketch] Consider an instance of two agents with weights $w_1 = \alpha$ and $w_2 = 1-\alpha$, and suppose there are $m=4$ goods (denoted $a_i$ for $1\leq i \leq 4)$ with the following valuation profiles:
\begin{center}
\begin{tabular}{ c | c c c c}
 	 & $a_1$ & $a_2$ & $a_3$ & $a_4$ \\ 
 	 \hline
	 $v_1$ & 0 & 1 & $\varphi$ & $\varphi$ \\  
 	 $v_2$ & 0 & 0 & 1 & 1\\  
\end{tabular}
\end{center}
where $\varphi := (1+\sqrt5)/2$. We seek to demonstrate that for all the possible allocations of these four items, the largest approximate factor obtainable is 0.786. More specifically, for each allocation we will derive the approximation ratio as a function of the item values and $\alpha$. After collecting these factors, we will adversarially select $\alpha$ to ensure that all are at most the desired approximate guarantee. The full proof of this result is deferred to the appendix due to space constraints.
\end{proof}

Furthermore, we expand this straightforward construction to the setting of $n=3$ agents. This extension serves to underscore the significant increase in complexity inherent in pursuing the WEFX objective, particularly when contrasted with its unweighted counterpart. The theorem's proof employs the same analytical approach as previously described and is thus deferred to the appendix.

\begin{theorem}\label{thm:wefx_n3}
For $n=3$ agents with additive valuation functions and weights $w_1,w_2$ and $w_3$ such that $w_1 + w_2 + w_3 = 1$, there exists an instance where a complete WEFX allocation of goods does not exist while a 0.795-WEFX allocation does.
\end{theorem}

As a consequence of these impossibility results, we proceed by designing procedures for the two-agent problem. These procedures yield a WEFX allocation or an approximation to this objective under varied problem assumptions.

\subsection{Two Agent Procedures}\label{sec:alg}
Our first algorithm computes a WEFX allocation for two agents with \emph{integer valued} weights. Formally, we are given two (additive) valuation functions, $v_1$ and $v_2$, for a set $\mathcal M$ of $m$ items wherein the corresponding agents have weights of 1 and $W \in \Z$. The algorithm involves a modified ``I-cut-you-choose" approach, where the agent with the higher weight greedily divides the goods into $W+1$ bundles. Agent 1 then gets the opportunity to select their preferred bundle from among these. The procedure is presented as pseudocode in Algorithm~\ref{alg:two_agent}.

\begin{algorithm}[t]
\caption{Integer Weight WEFX Algorithm} \label{alg:two_agent}
\begin{algorithmic}[1]
    \STATE Initialize $X_1 \leftarrow \emptyset, X_2 \leftarrow \emptyset$
    \IF{$m \le W$}
        \STATE $X_1 \leftarrow \arg\max_{g \in \M} v_1(g)$
        \STATE $X_2 \leftarrow \M \setminus X_1$
    \ELSE
        \STATE Initialize $(P_1, ..., P_{W+1}) \leftarrow (\emptyset, ..., \emptyset)$
        \WHILE{$\M \neq \emptyset$}
            \STATE $k = \arg\min_{i \in [W+1]} v_2(P_i)$
            \STATE $P_k \leftarrow P_k \cup \{g \in \arg\max_{h \in \M}v_2(h)\}$
        \ENDWHILE
        \STATE $X_1 = \arg\max_{k \in [W+1]} v_1(P_k)$
        \STATE $X_2 = \M \setminus X_1$
    \ENDIF
    \RETURN{$(X_1, X_2)$}
\end{algorithmic}
\end{algorithm}

\begin{theorem} \label{thm:int_wefx}
    For $n=2$ agents with weights 1 and $W$, along with additive valuation functions, Algorithm~\ref{alg:two_agent} computes a WEFX allocation.
\end{theorem}
\begin{proof}[Proof sketch]
    Due to space constraints, we here give a high-level intuition of the algorithms correctness and defer the full analysis to the appendix. 

    Given a set of $W+1$ bundles, agent 1's favorite bundle must necessarily be of value at least the average for the partition which necessarily guarantees the WEFX property. The crux of the analysis is thus verifying that, regardless of the first agent's selection, the second agent remains satisfied. By noting that the greedy partitioning procedure conducted by agent 2 reduces to computing an EFX allocation on $n=W+1$ identical agents, we can verify that the stronger notion of WEFX is necessarily satisfied.
\end{proof}

Despite this positive result, the inherent upper bound demonstrated in Theorem~\ref{thm:wefx_n2} indicates that adjusting the assumption of agent weightings to be \emph{normalized} instead compels us to design an approximation algorithm. We proceed to present Algorithm~\ref{alg:approx} which guarantees an \emph{approximate} WEFX allocation for two agents. The algorithm is a modification of the famous moving-knife procedure. Intuitively, the procedure works to minimally satisfy the higher priority agent before reassigning some items to the other agent in an effort to more equitably balance the two and mitigate any envious relationship. This natural algorithm adapted from the unweighted problem setting yields the first approximate WEFX procedure to date. Assuming $w$ is the weight of the higher priority agent, 
$\alpha = \frac{w}{2\sqrt[3]{m}}$ is our objective approximation factor of WEFX allocation. Our main technical result is stated formally as follows.

\vspace{1mm}
\begin{theorem} \label{thm:alg}
Algorithm~\ref{alg:approx} guarantees an $\frac{w}{2\sqrt[3]{m}}$-WEFX allocation for two agents.
\end{theorem}

In this algorithm, we first split the items into two bundles $\A_1$ and $\A_2$ based on their normalized valuations for the agents (ie. without loss of generality assume the sum of each agents valuation over all the items is 1). If the first agent has higher value for an item compared to the second agent, the item goes to the first bundle, otherwise it goes to the second bundle. Without loss of generality, we can assume that $v_1(\A_1) \geq w$. Specifically, if $v_1(\A_1) < w$ then we can re-index the two agents so that the inequality holds (since the inequality must hold for the other agent by normalization assumptions on the weights and valuations). We then sort the items in $\A_1$ in descending order based on the valuation function of the first agent, and assume that $g_i$ is the $i^{th}$ highest value item in $\A_1$ for the first agent. Also, assume that $k$ is the maximum number where $v_1(\{g_1, \ldots, g_{k-1}\}) \leq w$ holds. Then, we consider four different allocations as explained in the algorithm, and prove that at least one of these allocations must guarantee $\alpha$-WEFX.

\begin{algorithm}[tb]
\caption{Approximate WEFX Algorithm}
\label{alg:approx}
\begin{algorithmic}[1] 
\STATE \textbf{Input}: Two agents with valuation functions $v_1$ and $v_2$ with weights $w$ and $1-w$ respectively and set of items $\M$\\
\STATE \textbf{Output}: Allocation $\A = \{\A_1,\A_2\}$
\STATE Normalize the valuations, $v_1(\M) = v_2(\M) = 1$
\STATE Let $\A_1 = \{g \in \M | v_1(g) \geq v_2(g)\}$, $\A_2 = \M \setminus \A_1$
\STATE Sort $\A_1$ in descending, let $g_i$ be the $i$-th largest item 
\STATE Let $k = \arg \max_m v_1(\{g_1, \ldots, g_{m-1}\}) \leq w$
\STATE \textbf{Case \rom{1}}: Set $\A_1 = \{g_1\}$, $\A_2 = \M \setminus \A_1$
\IF {$\A$ is $\alpha$-WEFX}
    \STATE return $\A$.
\ENDIF
\STATE \textbf{Case \rom{2}}: Let $\A_1 = \{g_1,...,g_{k-1}\}$, $\A_2 = \M \setminus \A_1$

\IF {$\A$ is $\alpha$-WEFX}
    \STATE return $\A$.
\ENDIF
\STATE \textbf{Case \rom{3}}: Let $\A_1 = \{g_1,...,g_{k}\}$, $\A_2 = \M \setminus \A_1$
\IF {$\A$ is $\alpha$-WEFX}
    \STATE return $\A$.
\ENDIF
\STATE \textbf{Case \rom{4}}: Let $\A_2 = \arg \max_{g \in B_1} v_2(g)$, $\A_1 = \M \setminus \A_2$
\end{algorithmic}
\end{algorithm}

We proceed to prove the result by analyzing each case independently to obtain inequalities that would necessarily hold and lastly show that all cannot hold at the same time, thus yielding our result. To start, in Case \rom{1}, we check if we can satisfy the first agent by allocating only one item to her. Hence, the second agent cannot envy this agent after removing the minimum item from $\A_1$. If this does not guarantee an $\alpha$-WEFX allocation, we can prove the following lemma:

\begin{lemma} \label{lemma: ineq1}
If Case \rom{1} fails, we have:
\begin{equation} \label{case1ineq}
    1/k < \alpha \left(\frac{1 - w/k}{1 - w}\right)
\end{equation}
\end{lemma}

\begin{proof}
According to the definition of $k$, $v_1(\{g_1, \ldots, g_{k}\} \ge w$. Hence, we have $v_1(g_1) > w/k$. This further implies that $v_1(\A_2) < 1 - w/k$. If $\A$ is $\alpha$-WEFX, we return this allocation. Otherwise, the first agent envies the second even after adding the approximation factor, thus we have:

$$\frac{w/k}{w} < \alpha \left(\frac{1 - w/k}{1 - w}\right)$$
Simplification yields (\ref{case1ineq}).  
\end{proof}

\noindent Let $v_1(\{g_1, \ldots, g_{k-1}\} = w - \beta$. Lemma~\ref{lemma: beta} proves that $\beta$ is bounded by $w/k$.

\begin{lemma} \label{lemma: beta}
If $v_1(\{g_1, \ldots, g_{k-1}\} = w - \beta$, then $0 \leq \beta < w/(k-1)$.
\end{lemma}

\begin{proof}
By definition, $\beta \geq 0$ necessarily holds. Assume towards contradiction that $\beta \geq w/(k-1)$. This yields $v_1(g_k) > w/(k-1)$. Since the items of $\A_1$ are sorted, we further have that $v_1(g_i) > w/(k-1)$ for every $1 \leq i < k$. Hence, $v_1(\{g_1, \ldots, g_{k-1}\} > w$, a contradiction.  
\end{proof}

If Case \rom{2} occurs, an $\alpha$-WEFX allocation is guaranteed, and the problem is solved. Otherwise, we can prove the following lemma: 

\begin{lemma} \label{lemma: ineq2}
If Case \rom{2} fails, we have:
\begin{equation} \label{case2ineq}
    \frac{w - \beta}{w} < \alpha \left(\frac{1 - w + \beta}{1 - w}\right)
\end{equation}
\end{lemma}

\begin{proof}
In Case \rom{2}, we allocate the first $k-1$ items of $\A_1$ to the first agent, and we have $v_1(\A_1) = w - \beta$. Since these items come from $\A_1$, by definition we have $v_2(\A_1) \leq w - \beta$. Hence, we have $v_2(\A_2) \geq 1 - w + \beta$. Therefore, if we normalize the allocations based on the weights $w$ and $1-w$, the second agent cannot envy the first agent. If $\A$ is $\alpha$-WEFX, this is the final allocation, otherwise, the first agent envies the second even after adding the approximation factor, and therefore Inequality (\ref{case2ineq}) holds.  
\end{proof}

If Case \rom{3} occurs, an $\alpha$-WEFX allocation is guaranteed, and the problem is solved. Otherwise, we have the following result: 

\begin{lemma} \label{lemma: ineq3}
If Case \rom{3} fails, we have:
\begin{equation} \label{case3ineq}
    \frac{1 - w - w/(k-1) + \beta}{1 - w} < \alpha \left(\frac{w - \beta}{w}\right)
\end{equation}
\end{lemma}

\begin{proof}
In Case \rom{3}, we allocate the first $k$ items of $\A_1$ to the first agent, and have $v_1(\A_1) > w$. In this case, if we normalize the allocations based on the weights $w$ and $1-w$, the first agent does not envy the second agent. Furthermore, we have $v_1(\{g_1, \ldots, g_{k-1}\} = w - \beta$, and as we discussed in the proof of Lemma~\ref{lemma: beta}, this implies $v_1(g_k) \leq w/(k-1)$. Hence, we can say that $v_1(\A_1) \leq w - \beta + w/(k-1)$. Since the value of every item in $\A_1$ for the first agent is \emph{at least} its value for the second agent, we can say that $v_2(\A_1) \leq w - \beta + w/(k-1)$ and therefore $v_2(\A_2) \geq 1 - w + \beta - w/(k-1)$.

Since for every item in $\A_1$ the value of the first agent is at least the value of the second agent, and given the fact that $v_1(\A_1)$ after removing the minimum item from $\A_1$ (according to the first agent's perspective) is equal to $w - \beta$, we thus argue that $v_2(\A_1)$ is at most $w - \beta$ after removing the minimum item from $v_1(\A_1)$ (according to the second agent's perspective). 

If $\A$ is $\alpha$-WEFX, this is the final allocation, otherwise, the second agent envies  the first even after adding the approximation factor, and as a result, Inequality (\ref{case3ineq}) should hold.  
\end{proof}

If none of these three cases occur, we must have Case \rom{4}. If the allocation is $\alpha$-WEFX, the problem is solved. Otherwise, the following inequality bound must hold: 

\begin{lemma} \label{lemma: ineq4}
If Case \rom{4} occurs but the allocation is not $\alpha$-WEFX, we have: 
\begin{equation} \label{case4ineq}
    \frac{w/k}{1 - w} < \alpha \left(\frac{1 - w/k}{w}\right)
\end{equation}
\end{lemma}

\begin{proof}
In Case \rom{4}, we allocate the item with the highest value among all the items in $\A_1$ to the second agent. Since Case \rom{3} failed, we have $v_2(\{g_1, \ldots, g_{k}\}) > w$. Hence, the value of at least one of the items in $\{g_1, \ldots, g_{k}\}$ is at least $w/k$ for the second agent. Therefore, in the Case \rom{4} allocation, $v_2(\A_2) \geq w/k$ and $v_2(\A_1) \leq 1 - w/k$.  

If $\A$ is $\alpha$-WEFX, this is the final allocation, otherwise, the second agent envies the first agent even after adding the approximation factor, and as a result, Inequality (\ref{case4ineq}) should hold. (Since we only allocate one item to the second agent, the first agent does not envy the second after removing the minimum item from $\A_2$.)  
\end{proof}

We now have the necessary tools to prove our main theorem which states that Algorithm~\ref{alg:approx} guarantees an $\alpha$-WEFX allocation. We show this by proving that Inequalities (\ref{case1ineq}-\ref{case4ineq}) cannot hold at the same time.

\begin{proof}[Proof of Theorem~\ref{thm:alg}]
If Algorithm~\ref{alg:approx} does not guarantee an $\alpha$-WEFX allocation, then all the Inequalities (\ref{case1ineq} - \ref{case4ineq}) should hold. 
After simplifying Inequality (\ref{case1ineq}), we have 
$$1 - w < \alpha (k - w)$$
Since $0 < \alpha < 1$, we necessarily have $k > 1$. Hence, we can safely use $k-1$ as a denominator in Inequality (\ref{case3ineq}).
Also after merging Inequalities (\ref{case2ineq}) and (\ref{case3ineq}), we have: 
    $$1 - w - w/(k-1) + \beta < \alpha^2 (1 - w + \beta)$$
Since we have $0 < \alpha < 1$, the following inequality should hold:
    $$1 - w - w/(k-1) < \alpha^2 (1 - w)$$
Since $k \geq 2$, we can use the above inequality and write:
\begin{equation} \label{ineqApprox1}
    1 - w - 2w/k < \alpha^2 (1 - w)
\end{equation}
After simplifying Inequality (\ref{case4ineq}), we have:
\begin{equation}\label{ineqApprox2}
    (w^2/k) (1/\alpha - 1) + w/k + w < 1
\end{equation}
We rewrite Inequality~\ref{ineqApprox1} using ~\ref{ineqApprox2} as follows:
$$(w^2/k) (1/\alpha - 1) + w/k + w - w - 2w/k < \alpha^2 (1 - w)$$
After simplifying the above inequality, we have:
$$w^2 < \alpha^3 k (1-w) + \alpha w (w+1) $$
Since $0 < w < 1$, we can rewrite the above inequality as follows:
$$w^2 < \alpha^3 k + 2 \alpha w $$
If we use $\alpha = \frac{w}{2\sqrt[3]{m}}$, we have:
$$w^2 < \frac{w^3 k}{8m} + \frac{w^2}{\sqrt[3]{m}}$$
Since $0 < w < 1$ and $2 \leq k \leq m$, we have:
$$w^2 < \frac{w^2}{8} + \frac{w^2}{\sqrt[3]{2}}$$
which is a contradiction. Hence, all the Inequalities (\ref{case1ineq} - \ref{case4ineq}) cannot hold at the same time, proving the main result.  
\end{proof}

\section{Weighted Division of Indivisible Chores}
We now turn attention to the problem of chore division. Specifically, we present symmetric upper bound results for envy-freeness up to any chore to those of Section~\ref{sec:imposs} before demonstrating that a variant on the WEF1 algorithm of \cite{Chakraborty_2020} yields weighted envy-freeness up to one chore.
\subsection{General Impossibility} \label{sec:xwef}
As is the case for the allocation of goods, we cannot in general guarantee that an allocation which is envy-free up to any chore exists for even small problem instances. By slight modification, the construction used for Theorem~\ref{thm:wefx_n2} yields the following non-existence result.

\begin{theorem}\label{thm:xwef_n2}
For $n=2$ agents with additive cost functions and weights $w_1$ and $w_2$ such that $w_1 + w_2 = 1$, there exists an instance where a complete XWEF allocation of chores does not exist while a 1.272-XWEF allocation does.
\end{theorem}

Note that the approximation factor here of 1.272 is equivalent to the reciprocal of the approximation factor for WEFX, as the two problems are nearly symmetric to one another. Lastly, this counterexample construction can be extended to the $n=3$ setting:

\begin{theorem}\label{thm:xwef_n3}
For $n=3$ agents with additive cost functions and weights $w_1,w_2$ and $w_3$ such that $w_1 + w_2 + w_3 = 1$, there exists an instance where a complete XWEF allocation of chores does not exist while a 1.214-XWEF allocation does.
\end{theorem}

\subsection{1WEF Algorithm}\label{sec:1wef}
We lastly present our polynomial time algorithm for weighted envy-freeness up to one chore (1WEF) to compliment the WEF1 algorithm of \cite{Chakraborty_2020}. To achieve this, we must devise a proper sequential picking protocol that allows agents to pick their most preferred chores in a predetermined ordering in a similar vein to the traditional \emph{round-robin} procedure from the symmetric agent literature \cite{Budish_2010}. 

In our general case with unequal weights, we devise a \emph{weight-dependent} picking sequence, 
in addition to an arbitrary ordering final loop, to yield the desired 1WEF allocation for any number of agents and arbitrary weights. Although the proof of our algorithm is intricate and precise, the algorithm itself is intuitive and computationally efficient.
\begin{theorem}\label{thm:1wef}
For any number of agents with additive cost functions and arbitrary positive real weights, there exists an algorithm that computes a 1WEF complete allocation in polynomial time.
\end{theorem}
The crux of this theorem relies on the carefully constructed picking-sequence of Algorithm~\ref{alg:1wef} so that each agent receives at least one item of higher cost (in the final loop) and the remaining items are allocated based on a weight-adjusted picking frequency for each agent. We claim that every agent is 1WEF up to the chore selected in the final for loop at every iteration of the while loop. We begin by presenting two insightful lemmas concerning the algorithm itself, their proofs are deferred to the appendix due to space constraints. 

\begin{algorithm}[tb]
\caption{Chore-Division Round-Robin}
\label{alg:1wef}
\begin{algorithmic}[1] 
\STATE \textbf{Input}: $\mathcal{L}, (w_1,w_2,...,w_n), c_i \text{ for each } i \in \mathcal{N}$\\
\STATE \textbf{Output}: returns 1WEF allocation for the $n$ agents.
\STATE Remaining chores $\widehat{\mathcal{L}} \leftarrow \mathcal{L}$
\STATE Bundles $B \leftarrow \emptyset, \forall i \in \N$
\STATE $t_i \leftarrow 1, \forall i \in \N$
\STATE \textit{Allocate chores:}
\WHILE {$\widehat{|\mathcal{L}}| > |\N|$}
	\STATE $i^{*} \leftarrow \arg\min_{i \in \N} \frac{t_i}{w_i}$ , break ties lexographically
	\STATE $l^{*} \leftarrow \arg\min_{l \in \widehat{\mathcal{L}}} c_{i^{*}}(l)$
	\STATE $B_{i^{*}} \leftarrow B_{i^{*}} \cup \{l^{*}\}$
	\STATE $\widehat{\mathcal{L}} \leftarrow \widehat{\mathcal{L}} \setminus l^*$
	\STATE $t_{i^{*}} \leftarrow t_{i^{*}} + 1$
\ENDWHILE
\STATE \textit{Allocate Remaining Chores:}
\FOR {$i \in \N$} 
	\STATE $l^{*} \leftarrow \arg\min_{l \in \widehat{\mathcal{L}}} c_i(l)$
	\STATE $B_i \leftarrow B_i \cup \{l^{*}\}$
	\STATE $\widehat{\mathcal{L}} \leftarrow \widehat{\mathcal{L}} \setminus l^*$
\ENDFOR
\end{algorithmic}
\end{algorithm}

\begin{lemma}\label{thm:1wef_lem1}
Consider an agent $i$ chosen by Algorithm~\ref{alg:1wef} to pick a chore at some iteration t, and suppose it is not their first pick. Let $t_i$ and $t_j$ denote the number of times agents $i$ and $j$ picked a chore respectively prior to the current iteration. Then $\frac{t_j}{t_i} \geq \frac{w_j}{w_i}$.
\end{lemma}

Lemma~\ref{thm:1wef_lem1} is sufficient for ensuring that agent $i$ remains weighted envy-free up to the final chore allocated at each iteration of the loops execution. We verify this guarantee this fact using the following:
\begin{lemma}\label{thm:1wef_lem2}
Suppose that, for every iteration $t$ in which agent $i$ picks an item prior to their final pick, the number of times that agents $i$ and $j$ have picked chores ($t_i$ and $t_j$ respectively) satisfy $\frac{t_j}{t_i} \geq \frac{w_j}{w_i}$. Then, in every partial allocation (plus the final chore of the for loop) up to and including $i$'s latest pick, agent j is weighted envy-free up to the item allocated in the final loop of Algorithm~\ref{alg:1wef}.
\end{lemma}

Using Lemmas~\ref{thm:1wef_lem1} and \ref{thm:1wef_lem2}, we now have all the necessary facts to prove Theorem~\ref{thm:1wef}. Intuitively, we maintain an 1WEF invariant up to the chores that are allocated last, which will necessarily incur the most cost for each agent. This is effectively a reversal of the weighted picking sequence procedure of \cite{Chakraborty_2020} that maintains the invariant of a WEF1 allocation up the \emph{first} allocated item (of largest value to each agent). Though simplistic and intuitive, this extension is non-trivial and relies on a careful analysis of the initial weighted-picking sequence to ensure that our inductive invariant is maintained. Due to space constraints, we defer the reader to Appendix \ref{appendix:1wef} for this intricate analysis, and here present a proof of the algorithm's correctness using the two results depicted above.

\begin{proof}[Proof of Theorem~\ref{thm:1wef}]
In the instance that $n > l$, we have that only the for loop of Algorithm~\ref{alg:1wef} runs. Thus each agent is allocated \textit{at most} one chore, thus any pair of agents is 1WEF, and more so XWEF.

Now in the more interesting case of $l>n$, we invoke Lemmas~\ref{thm:1wef_lem1} and \ref{thm:1wef_lem2}. The combination of the two lemmas gives us that on any iteration when agent $i$ is picking a chore, they will remain 1WEF up to the final chore allocated in the second loop. As this property holds until a complete allocation is achieved, we have that the final result is 1WEF.

For time complexity of the algorithm, we note that the while loop runs exactly $l - n$ iterations. Within this loop, identification of the minimal $\frac{t_i}{w_i}$ takes at $O(n)$ time followed by the picking of the minimal cost chore, which takes $O(l)$ time. Since the first loop only runs in the event $l > n$, the loop runs in $O(l^2)$ time. Lastly, we have the for loop which runs in $O(n)$ iterations, each taking $O(1)$ time. Therefore, the algorithm runs in time $O(l^2)$.
\end{proof}

\section{Conclusions and Future Work}
Our work highlights the increased complexity of the fair allocation problem on both goods and chores when agents are assumed to be asymmetric. While strong negative existential results on WEFX (and XWEF) for the case of $n=2$ and 3 showcase a considerable deviation from the unweighted problem setting, our simple adaptation on the ``I-cut-you-choose'' and moving-knife procedures provide a novel first step on both exact and approximate guarantees for these challenging problems. Furthermore, the presented intricate analysis of our $\frac{w}{2\sqrt[3]{m}}$-WEFX suggests that the current methods for approximate EFX guarantees will not necessarily translate to this generalized context. For example, Proposition~\ref{prop:chakra} reveals the inadequacy of the envy-cycle procedure for general valuations. Furthermore, our efforts to ensure a 1WEF allocation highlight that the standard round-robin procedure requires specific adjustments to encompass the asymmetric agent configuration. Consequently, the amalgamation of these two approaches, though effective in yielding a $(\varphi - 1)$-EFX allocation \cite{Amanatidis_2020}, does not naturally expand to address this extended problem.


We leave the question of the existence of approximate WEFX allocations open as well as the connection between WEFX and the other studied fairness notions for the weighted context for future work. We hope the present work inspires more study on the generalized notion of weighted fair division problems as they will require a plethora of novel techniques to tackle.

\section*{Acknowledgements}
This work is partially supported by DARPA QuICC NSF AF:Small \#2218678, and NSF AF:Small
\#2114269. Max Springer was supported by the National Science Foundation Graduate Research Fellowship Program under Grant No. DGE 1840340. Any opinions, findings, and conclusions or recommendations expressed in this material are those of the author(s) and do not necessarily reflect the views of the National Science Foundation.


\bibliography{aaai24}

\begin{thebibliography}{30}
\providecommand{\natexlab}[1]{#1}

\bibitem[{Amanatidis et~al.(2021)Amanatidis, Birmpas, Filos-Ratsikas,
  Hollender, and Voudouris}]{Amanatidis_2020}
Amanatidis, G.; Birmpas, G.; Filos-Ratsikas, A.; Hollender, A.; and Voudouris,
  A.~A. 2021.
\newblock Maximum Nash welfare and other stories about EFX.
\newblock \emph{Theoretical Computer Science}, 863: 69--85.

\bibitem[{Amanatidis, Markakis, and Ntokos(2020)}]{Amanatidis_2020_2}
Amanatidis, G.; Markakis, E.; and Ntokos, A. 2020.
\newblock Multiple birds with one stone: Beating 1/2 for EFX and GMMS via envy
  cycle elimination.
\newblock \emph{Theoretical Computer Science}, 841: 94–109.

\bibitem[{Aziz, Chan, and Li(2019)}]{Aziz_2019_2}
Aziz, H.; Chan, H.; and Li, B. 2019.
\newblock Weighted maxmin fair share allocation of indivisible chores.
\newblock \emph{arXiv preprint arXiv:1906.07602}.

\bibitem[{Aziz, Moulin, and Sandomirskiy(2020)}]{aziz2020polynomial}
Aziz, H.; Moulin, H.; and Sandomirskiy, F. 2020.
\newblock A polynomial-time algorithm for computing a Pareto optimal and almost
  proportional allocation.
\newblock \emph{Operations Research Letters}, 48(5): 573--578.

\bibitem[{Babaioff, Ezra, and Feige(2021)}]{Babaioff_2021}
Babaioff, M.; Ezra, T.; and Feige, U. 2021.
\newblock Fair-share allocations for agents with arbitrary entitlements.
\newblock \emph{arXiv preprint arXiv:2103.04304}.

\bibitem[{Babaioff, Nisan, and Talgam-Cohen(2019)}]{babaioff2019fair}
Babaioff, M.; Nisan, N.; and Talgam-Cohen, I. 2019.
\newblock Fair Allocation through Competitive Equilibrium from Generic Incomes.
\newblock In \emph{FAT}, 180.

\bibitem[{Barbanel(1996)}]{barbanel1996game}
Barbanel, J. 1996.
\newblock Game-theoretic algorithms for fair and strongly fair cake division
  with entitlements.
\newblock In \emph{Colloquium Mathematicae}, volume~69, 59--73.

\bibitem[{Barman and Krishnamurthy(2020)}]{Barman_2020}
Barman, S.; and Krishnamurthy, S.~K. 2020.
\newblock Approximation algorithms for maximin fair division.
\newblock \emph{ACM Transactions on Economics and Computation (TEAC)}, 8(1):
  1--28.

\bibitem[{Bouveret, Chevaleyre, and Maudet(2016)}]{Bouveret_2016}
Bouveret, S.; Chevaleyre, Y.; and Maudet, N. 2016.
\newblock Fair Allocation of Indivisible Goods.

\bibitem[{Brams and Taylor(1995)}]{Brams_1995}
Brams, S.~J.; and Taylor, A.~D. 1995.
\newblock An Envy-Free Cake Division Protocol.
\newblock \emph{The American Mathematical Monthly}, 102(1): 9--18.

\bibitem[{Budish(2010)}]{Budish_2010}
Budish, E. 2010.
\newblock The combinatorial assignment problem: approximate competitive
  equilibrium from equal incomes.
\newblock In \emph{BQGT}.

\bibitem[{Chakraborty et~al.(2020)Chakraborty, Suksompong, Zick, and
  Igarashi}]{Chakraborty_2020}
Chakraborty, M.; Suksompong, W.; Zick, Y.; and Igarashi, A. 2020.
\newblock Weighted Envy-freeness in Indivisible Item Allocation.
\newblock \emph{ACM Transactions on Economics and Computation (TEAC)}, 9: 1 --
  39.

\bibitem[{Chaudhury, Garg, and Mehlhorn(2020)}]{chaudhury2020efx}
Chaudhury, B.~R.; Garg, J.; and Mehlhorn, K. 2020.
\newblock EFX exists for three agents.
\newblock In \emph{Proceedings of the 21st ACM Conference on Economics and
  Computation}, 1--19.

\bibitem[{Chevaleyre, Endriss, and Maudet(2017)}]{chevaleyre2017distributed}
Chevaleyre, Y.; Endriss, U.; and Maudet, N. 2017.
\newblock Distributed fair allocation of indivisible goods.
\newblock \emph{Artificial Intelligence}, 242: 1--22.

\bibitem[{Cseh and Fleiner(2020)}]{cseh2020complexity}
Cseh, {\'A}.; and Fleiner, T. 2020.
\newblock The complexity of cake cutting with unequal shares.
\newblock \emph{ACM Transactions on Algorithms (TALG)}, 16(3): 1--21.

\bibitem[{Dehghani et~al.(2018)Dehghani, Farhadi, HajiAghayi, and
  Yami}]{dehghani2018envy}
Dehghani, S.; Farhadi, A.; HajiAghayi, M.; and Yami, H. 2018.
\newblock Envy-free chore division for an arbitrary number of agents.
\newblock In \emph{Proceedings of the Twenty-Ninth Annual ACM-SIAM Symposium on
  Discrete Algorithms}, 2564--2583. SIAM.

\bibitem[{Farhadi et~al.(2019)Farhadi, Ghodsi, Hajiaghayi, Lahaie, Pennock,
  Seddighin, Seddighin, and Yami}]{Farhadi_2017}
Farhadi, A.; Ghodsi, M.; Hajiaghayi, M.~T.; Lahaie, S.; Pennock, D.; Seddighin,
  M.; Seddighin, S.; and Yami, H. 2019.
\newblock Fair Allocation of Indivisible Goods to Asymmetric Agents.
\newblock \emph{Journal of Artificial Intelligence Research}, 64: 1--20.

\bibitem[{Farhadi et~al.(2021)Farhadi, Hajiaghayi, Latifian, Seddighin, and
  Yami}]{Farhadi_2021}
Farhadi, A.; Hajiaghayi, M.; Latifian, M.; Seddighin, M.; and Yami, H. 2021.
\newblock Almost Envy-freeness, Envy-rank, and Nash Social Welfare Matchings.
\newblock In \emph{AAAI}.

\bibitem[{Feige, Sapir, and Tauber(2021)}]{Feige_2021}
Feige, U.; Sapir, A.; and Tauber, L. 2021.
\newblock A tight negative example for MMS fair allocations.
\newblock arXiv:2104.04977.

\bibitem[{Garg, McGlaughlin, and Taki(2019)}]{Garg_2019}
Garg, J.; McGlaughlin, P.~C.; and Taki, S. 2019.
\newblock Approximating Maximin Share Allocations.
\newblock In \emph{SOSA}.

\bibitem[{Lipton et~al.(2004)Lipton, Markakis, Mossel, and
  Saberi}]{Lipton_2004}
Lipton, R.; Markakis, E.; Mossel, E.; and Saberi, A. 2004.
\newblock On approximately fair allocations of indivisible goods.
\newblock In \emph{EC '04}.

\bibitem[{Markakis(2017)}]{markakis2017approximation}
Markakis, E. 2017.
\newblock Approximation algorithms and hardness results for fair division with
  indivisible goods.
\newblock \emph{Trends in Computational Social Choice}, 231--247.

\bibitem[{Peterson and Su(2002)}]{Peterson_2002}
Peterson, E.; and Su, F. 2002.
\newblock Four-Person Envy-Free Chore Division.
\newblock \emph{Mathematics Magazine}, 75: 117 -- 122.

\bibitem[{Peterson and Su(2009)}]{Peterson_2009}
Peterson, E.; and Su, F. 2009.
\newblock N-person envy-free chore division.
\newblock \emph{arXiv: Combinatorics}.

\bibitem[{Pikhurko(2000)}]{Pikhurko_2000}
Pikhurko, O. 2000.
\newblock On Envy-Free Cake Division.
\newblock \emph{The American Mathematical Monthly}, 107: 736 -- 738.

\bibitem[{Plaut and Roughgarden(2018)}]{Plaut_2018}
Plaut, B.; and Roughgarden, T. 2018.
\newblock Almost Envy-Freeness with General Valuations.
\newblock In \emph{SODA}.

\bibitem[{Procaccia(2009)}]{Procaccia_2009}
Procaccia, A.~D. 2009.
\newblock Thou Shalt Covet Thy Neighbor's Cake.
\newblock In \emph{IJCAI}.

\bibitem[{Robertson and Webb(1998)}]{Robertson_1998}
Robertson, J.; and Webb, W. 1998.
\newblock Cake-cutting algorithms - be fair if you can.

\bibitem[{Segal-Halevi and Suksompong(2020)}]{segal2020cut}
Segal-Halevi, E.; and Suksompong, W. 2020.
\newblock How to cut a cake fairly: A generalization to groups.
\newblock \emph{The American Mathematical Monthly}, 128(1): 79--83.

\bibitem[{Wu, Zhang, and Zhou(2023)}]{wu2023weighted}
Wu, X.; Zhang, C.; and Zhou, S. 2023.
\newblock Weighted EF1 Allocations for Indivisible Chores.
\newblock arXiv:2301.08090.

\end{thebibliography}

\newpage
\onecolumn
\appendix
\section{Omitted Proofs}
\subsection{Indivisible Goods}
\subsubsection{Impossibility Results}
\begin{proof}[Proof of Theorem~\ref{thm:wefx_n2}] First, we consider the simplest case where all goods are allocated to one of the agents, denote the allocation $\A_1 = \{a_1,a_2,a_3,a_4\}$:
\begin{align*}
	 \frac{v_2(\varnothing)}{1-\alpha} = 0 < \frac{2}{\alpha} = \frac{v_2(\A_1)}{\alpha}
\end{align*}
Which is by definition an envious relationship. We can also see that it is not WEFX by examining the value of the bundle after removing the good of minimal value and once again checking for envy, as this would guarantee that the removal of some item alleviates the envy if possible:
\begin{align*}
	\frac{v_2(\varnothing)}{1-\alpha} = 0 < \frac{2}{\alpha} = \max_{a \in \A_1}\frac{v_2(\A_1 \setminus \{a\})}{\alpha}
\end{align*}
It is also clear that by the additivity of the valuations, $$\max_{a \in \A_i}\frac{v_i(\A_j \setminus \{a\})}{w_j} \leq \frac{v_i(\A_j)}{w_j}.$$ Thus we need only check this condition when verifying an allocation is not WEFX. 

Converse to the above allocation, consider the case where all the goods are allocated to agent 2, denoting the allocation $\A_2 = \{a_1,a_2,a_3,a_4\}$. We once again see this yields an envious relationship which is not WEFX:
\begin{gather*}
	 \frac{v_1(\varnothing)}{\alpha} = 0 < \frac{2\phi + 1}{1-\alpha} = \max_{a \in \A_2}\frac{v_1(\A_2 \setminus \{a\})}{1-\alpha}
\end{gather*}
Now consider the case where one agent is allocated exactly one of the goods. The agent who receives the three goods will always be WEFX relative to the other, as removing the other agents single good will yield a value of zero. Thus, we need only consider the agent who receives exactly one good. We first evaluate the possibilities for agent 1:
\vspace{2.5mm} \\
\noindent (1) $\A_2 = \{a_1,a_2,a_3\} :$
$$\frac{v_1(\{a_4\})}{\alpha} = \frac{\phi}{\alpha} < \frac{\phi+1}{1-\alpha} = \max_{a \in \A_2}\frac{v_1(\A_2 \setminus \{a\})}{1-\alpha}
		\Rightarrow \frac{\phi(1-\alpha)}{\alpha(\phi+1)}\text{-WEFX}$$

\noindent (2) $\A_2 = \{a_1,a_3,a_4\} :$
$$\frac{v_1(\{a_2\})}{\alpha} = \frac{1}{\alpha} < \frac{2\phi}{1-\alpha} = \max_{a \in \A_2}\frac{v_1(\A_2 \setminus \{a\})}{1-\alpha}
		\Rightarrow \frac{1-\alpha}{2\phi\alpha}\text{-WEFX}$$

\noindent (3) $\A_2 = \{a_2,a_3,a_4\} :$
$$\frac{v_1(\{a_1\})}{\alpha} = 0
		\Rightarrow 0\text{-WEFX}$$
We omit the allocation $\A_2 = \{a_1,a_2,a_4\}$ as this yields an equivalent valuation inequality to (1). Next we consider the converse case of allocating exactly one good to the second agent.
\vspace{2.5mm} \\
\noindent (4) $\A_1 = \{a_1,a_2,a_3\} :$
$$\frac{v_2(\{a_4\})}{1-\alpha} = \frac{1}{1-\alpha} < \frac{1}{\alpha} = \max_{a \in \A_1}\frac{v_2(\A_1 \setminus \{a\})}{\alpha}
		\Rightarrow \frac{\alpha}{1-\alpha}\text{-WEFX}$$

\noindent (5) $\A_1 = \{a_1,a_3,a_4\} :$
$$\frac{v_2(\{a_2\})}{1-\alpha} = 0
		\Rightarrow 0\text{-WEFX}$$
We omit here the cases where agent 2 is allocated only $a_1$ or $a_3$ as these have equivalent valuations to the above. Lastly, we consider the case where each agent is allocated exactly two goods.
\vspace{2.5mm} \\
\noindent (6) $\A_1 = \{a_1,a_2\}, \A_2 = \{a_3,a_4\}  :$
$$\frac{v_1(\A_1)}{\alpha} = \frac{1}{\alpha} < \frac{\phi}{1-\alpha} = \max_{a \in \A_2}\frac{v_1(\A_2 \setminus \{a\})}{1-\alpha}
		\frac{1-\alpha}{\alpha \cdot \phi}\text{-WEFX}$$

\noindent (7) $\A_1 = \{a_1,a_3\}, \A_2 = \{a_2,a_4\}  :$
$$\frac{v_2(\A_2)}{1-\alpha} = \frac{1}{1-\alpha} < \frac{1}{\alpha} = \max_{a \in \A_1}\frac{v_2(\A_1 \setminus \{a\})}{\alpha}
		\frac{\alpha}{1-\alpha}\text{-WEFX}$$

\noindent (8) $\A_1 = \{a_3,a_4\}, \A_2 = \{a_1,a_2\}  :$
$$\frac{v_2(\A_2)}{1-\alpha} = 0
		\Rightarrow 0\text{-WEFX}$$
Finally, noting that $\A_1 = \{a_2,a_3\},\{a_1,a_4\}\text{, and }\{a_2,a_4\}$ all result in the same valuations and are thus omitted. We thus conclude that no WEFX allocation exists for the given scenario since all possible unique allocations to the two agents with associated valuation profiles have been examined with none meeting the defined criteria.  
\vspace{2.5mm} \\
\noindent It is also important to examine the maximal $c$-approximate WEFX allocation achieved in the above counterexample:
\begin{gather*}
	\max_{\alpha \in (0,1)} \text{ } \{\frac{\alpha}{1-\alpha},\frac{1-\alpha}{\alpha \cdot \phi},\frac{\phi(1-\alpha)}{\alpha(\phi+1)}\}
\end{gather*}
We finally obtain the maximal approximate factor by equating the three expressions and solving to get $\alpha = 0.44$. Therefore, our maximal approximate factor is 0.786-WEFX.
\end{proof}

\begin{proof}[Proof of Theorem~\ref{thm:wefx_n3}] For the given problem instance:
\begin{center}
\begin{tabular}{ c | c c c c c}
 	 & $a_1$ & $a_2$ & $a_3$ & $a_4$ & $a_5$ \\ 
 	 \hline
	 $v_1$ & 0 & 0 & 0 & 0 & 1 \\  
 	 $v_2$ & $\phi$ & $\phi$ & 1 & 0 & 0 \\
 	 $v_3$ & 1 & 1 & 0 & 0 & 1
\end{tabular}
\end{center}
with $w = (\frac{3}{19},\frac{7}{19},\frac{9}{19})$ we have that the allocation to the three agents with a highest approximate factor of 0.795 is
$$\A_1 = \{a_5\}, \A_2 = \{a_1\}, \A_3 = \{a_2,a_3,a_4\}$$
\end{proof}

\subsubsection{WEFX Procedures}
\begin{proof}[Proof of Theorem~\ref{thm:add_wefx}]
    Let $v$ denote the valuation function for all agents. By construction of the algorithm \cite{Lipton_2004}, the (incomplete) allocation at the end of each iteration is guaranteed to be WEFX as long as we can find an agent, say $i$, towards whom no other agent has weighted envy at the beginning of the iteration. This agent will pick their favorite among the remaining items which will then be the least valuable item in $i$'s bundle according to all agents due to the identical valuations (and the greedy selection by $i$). At this step in the allocation procedure, we give the item under consideration to agent $i$ and thus any resulting weighted envy towards $i$ can be eliminated by removing this item. If there is no unenvied agent, then the weighted envy graph consists of at least one cycle; however, under identical valuations, the envy graph cannot have cycles. To see this, suppose that agents $1,2,...,\ell$ form a cycle (in that order) for some $\ell \in [n]$. Since agents have identical valuations, it must be that 
    \begin{align*}
        \frac{v(\A_1)}{w_1} < \frac{v(\A_2)}{w_2} < ... < \frac{v(A_\ell)}{w_\ell} < \frac{v(\A_1)}{w_1},
    \end{align*}
    which yields a contradiction.
\end{proof}

\begin{proof}[Proof of Theorem~\ref{thm:ord_wefx}]
    We now consider the slightly relaxed model of identical ordinal preferences: there exists some ordering of the set, denoted $\sigma \in S_n$, such that $v_i(\sigma_1) \ge v_i(\sigma_2) \ge ... \ge v(\sigma_m)$ for all agents $i \in [n]$.

    Again by the envy-cycle elimination procedure, we always allocate to a source node in the envy-graph who naturally picks their favorite item among the remaining options. Since all agents have an identical ordering of the items, this is the favorite item for all agents and is moreover the least valuable within agent $i$'s bundle. Thus, any new envy incurred can be alleviated by removal of this item which implies WEFX as above.

    If there is no unenvied agent (source node in the graph), then there must be a cycle in the envy-graph. Thus, suppose that agents $1,2,...,\ell$ form a cycle (in that order) for some $\ell \in [n]$, i.e.
    \begin{align*}
        \frac{v_i(\A_i)}{w_i} < \frac{v_i(\A_{i+1})}{w_{i+1}} \text{ for $i$ under modulo $\ell$.}
    \end{align*}
    Let $i'$ be the agent of minimal weight in this cycle. By nature of the envy cycle, we must then have that
    \begin{align*}
        \frac{v_{i'}(\A_{i'})}{w_{i'}} < \frac{v_{i'}(\A_{i'+1})}{w_{i'+1}} \Rightarrow v_{i'}(\A_{i'}) < \frac{w_{i'}}{w_{i'+1}} v_{i'}(\A_{i'+1}) < v_{i'}(\A_{i'+1})
    \end{align*}
    and therefore the decycling procedure will strictly improve the allocation according to agent $i'$. We can thus repeat the process until the minimal weight agent is removed from the cycle. Since there always exists an agent of minimal weight, we can always remove an agent from the cycle until a source node exists in the envy-graph. This concludes the proof of WEFX.
\end{proof}

\begin{proof}[Proof of Theorem~\ref{thm:int_wefx}]
We first analyze the instance where $m \le w$. Naturally, since agent 1 is allocated only one item (denoted $g$), the second is WEFX (up to removal of this item). Therefore, we need only analyze the first agent's perspective. By selection $v_1(g) \ge v_1(h)$ for all $h \in \mathcal M$ and therefore $$v_1(g) \ge \frac{v_1(\mathcal M)}{m} \ge \frac{v_1(\mathcal M)}{w} \ge \frac{v_1(\mathcal M \setminus g)}{w}.$$ Thus, the allocation is WEFX.

    We now assume that $w > m$. Again, the first agent receives their favorite of the $W+1$ created bundles. Therefore, we must have that
    \begin{align*}
        v_1(\M \setminus X_1) \le W \cdot v_1(X_1)
    \end{align*}
    and dividing through by $W$ gives us that the agent is WEF (and, moreover, WEFX).

    The less-trivial proof of WEFX from the second agent's perspective relies on a simple yet crucial lemma concerning the greedy partitioning procedure.
    \begin{lemma} \label{lem:part_wefx}
        For any pair of subsets, $P_i$ and $P_j$, created in Algorithm~\ref{alg:two_agent} such that $v_2(P_i) < v_2(P_j)$, it must hold that $$v_2(P_i) \ge v_2(P_j \setminus g) \text{ for all } g \in P_j.$$
    \end{lemma}
    We first demonstrate that this lemma gives us the final WEFX property before proceeding to prove the lemma's correctness. First observe that
    \begin{align*}
        \frac{v_2(\M \setminus X_1)}{W} &\ge \min_k v_2(P_k)
    \end{align*}
    Naturally, if $v_2(P_k) \ge v_2(X_1)$ for all such partitions given to the second agent, then we are done. Therefore, assume that $v_2(X_1) > \min_k v_2(P_k)$. By Lemma~\ref{lem:part_wefx}, this further implies that $$\min_k v_2(P_k) \ge v_2(X_1 \setminus g)$$ for any $g \in X_1$. Combining with the above inequality gives us the desired WEFX property.
\end{proof}

We now circle back to prove Lemma~\ref{lem:part_wefx} and complete the analysis. 
\begin{proof}[Proof of Lemma~\ref{lem:part_wefx}]
    First, if $|P_j| = 1$ then the result is vacuously true. Thus, we assume that $|P_j| > 1$ which implies that there is some iteration in the run of Algorithm~\ref{alg:two_agent} where this bundle has a lower value to the second agent than $P_i$. Let $t$ be the last iteration in which the $j$-th bundle receives an item and use superscript notation to denote the bundle at this time (ie. $P_j^t \cup \{g^t\} = P_j$). This further implies that $P_i^t \ge P_j^t = P_j \setminus g$, and note that by the greedy nature of our algorithm, this $g$ is of minimal value in $P_j$ to the second agent. Therefore, at this iteration the result holds and since no further items are allocated to the $j$-th bundle, we have the result.
\end{proof}

\subsection{Indivisible Chores}
\subsubsection{Impossibility Results}

\begin{proof}[Proof of Theorem~\ref{thm:xwef_n2}] 
Consider the following instance on 4 chores:
\begin{center}
\begin{tabular}{ c | c c c c}
 	 & $b_1$ & $b_2$ & $b_3$ & $b_4$ \\ 
 	 \hline
	 $c_1$ & 0 & 0 & 1 & 1 \\  
 	 $c_2$ & 0 & 1 & $\phi$ & $\phi$\\  
\end{tabular}
\end{center}

We proceed by case analysis in the same manner as above. First, we consider the simplest case where all chores are allocated to one of the agents, denote the allocation $\B_2 = \{b_1,b_2,b_3,b_4\}$:
\begin{align*}
	 \frac{c_2(\B_2)}{1-\alpha} = \frac{2\phi + 1}{1-\alpha} > 0 = \frac{c_2(\varnothing)}{\alpha}
\end{align*}
Which is by definition an envious relationship. We can also see that it is not XWEF by examining the value of the bundle after removing the good of minimal value and once again checking for envy, as this would guarantee that the removal of some item alleviates the envy if possible:
\begin{align*}
	\max_{b \in \B_2} \frac{c_2(\B_2)}{1-\alpha} = \frac{2\phi + 1}{1-\alpha} > 0 = \frac{c_2(\varnothing)}{\alpha}
\end{align*}
It is also clear that by the additivity of the valuations, $\max_{a \in \B_i}\frac{c_i(\B_j \setminus \{b\})}{w_j} \leq \frac{c_i(\B_j)}{w_j}$. Thus we need only check this condition when verifying an allocation is not XWEF. 

Converse to the above allocation, consider the case where all the chores are allocated to agent 1, denoting the allocation $\B_1 = \{b_1,b_2,b_3,b_4\}$. We once again see this yields an envious relationship which is not XWEF:
\begin{gather*}
	 \frac{c_1(\B_1)}{\alpha} = \frac{2}{\alpha} < 0 = \max_{a \in \B_2}\frac{c_1(\B_2 \setminus \{b\})}{1-\alpha}
\end{gather*}
Now consider the case where one agent is allocated exactly one of the goods. The agent who receives the one chore will always be XWEF relative to the other, as removing the other agents single chore will yield a value of zero. Thus, we need only consider the agent who receives exactly three chores. We first evaluate the possibilities for agent 1:
\vspace{2.5mm}\\
\noindent (1) $\B_1 = \{b_1,b_2,b_3\} :$
$$\max_{b \in \B_1}\frac{c_1(\B_1 \setminus \{b\})}{\alpha} = \frac{1}{\alpha} > \frac{1}{1-\alpha} = \frac{c_1(\B_2)}{1-\alpha} \Rightarrow \frac{1-\alpha}{\alpha}\text{-XWEF}$$

\noindent (2) $\B_1 = \{b_1,b_3,b_4\}:$
$$\max_{b \in \B_1}\frac{c_1(\B_1 \setminus \{b\})}{\alpha} = \frac{2}{\alpha} > 0 = \frac{c_1(\B_2)}{1-\alpha}
		\Rightarrow 0\text{-XWEF}$$
We omit here the cases where agent 2 is allocated only $b_1$ or $b_3$ as these have equivalent valuations to the above.
\vspace{2.5mm} \\
\noindent (3) $\B_2 = \{b_1,b_2,b_3\} :$
$$\max_{b \in \B_2}\frac{c_2(\B_2 \setminus \{b\})}{1-\alpha} = \frac{1 + \phi}{1-\alpha} > \frac{\phi}{\alpha} = \frac{c_2(\B_1)}{\alpha}
		\Rightarrow \frac{\phi(1+\phi)}{\alpha(1-\alpha)}\text{-XWEF}$$

\noindent (4) $\B_2 = \{b_1,b_3,b_4\} :$
$$\max_{b \in \B_2}\frac{c_2(\B_2 \setminus \{b\})}{1-\alpha} = \frac{2\phi}{1-\alpha} > \frac{1}{\alpha} = \frac{c_2(\B_1)}{\alpha}
		\Rightarrow \frac{2\phi\alpha}{1-\alpha}\text{-XWEF}$$

\noindent (5) $\B_2 = \{b_2,b_3,b_4\} :$
$$\max_{b \in \B_2}\frac{c_2(\B_2 \setminus \{b\})}{1-\alpha} = \frac{1+2\phi}{1-\alpha} > 0 = \frac{c_2(\B_1)}{\alpha}
		\Rightarrow 0\text{-XWEF}$$
We omit here the cases where agent 1 is allocated only $b_3$ as this is equivalent to (3). Lastly, we consider the case where each agent is allocated exactly two goods.
\vspace{2.5mm} \\
\noindent (6) $\B_1 = \{b_2,b_4\} , \B_2 = \{b_1,b_3\} :$
$$\max_{b \in \B_1} \frac{c_1(\B_1 \setminus \{b\})}{\alpha} = \frac{1}{\alpha} > \frac{1}{1-\alpha} = \frac{c_1(\B_2)}{1-\alpha}
		\frac{1-\alpha}{\alpha}\text{-XWEF}$$

\noindent (7) $\B_1 = \{b_1,b_2\}, \B_2 = \{b_3,b_4\} :$
$$\max_{b \in \B_2} \frac{c_2(\B_2 \setminus \{b\})}{1-\alpha} = \frac{\phi}{1-\alpha} > \frac{1}{\alpha} = \frac{c_2(\B_1)}{\alpha}
		\Rightarrow \frac{\phi\alpha}{1-\alpha}\text{-XWEF}$$

\noindent (8) $\B_1 = \{b_3,b_4\}, \B_2 = \{b_1,b_2\} :$
$$\max_{b \in \B_1}\frac{c_1(\B_1 \setminus \{b\})}{\alpha} = \frac{1}{\alpha} > 0 = \frac{c_1(B_2)}{1-\alpha}
		\Rightarrow 0\text{-XWEF}$$
Finally, noting that $\B_1 = \{b_1,b_3\},\{b_2,b_3\}\text{, and }\{b_1,b_4\}$ all result in the same valuations and are thus omitted. We thus conclude that no XWEF allocation exists for the given scenario since all possible unique allocations to the two agents with associated valuation profiles have been examined with none meeting the defined criteria.  
\vspace{2.5mm} \\
\noindent It is also important to examine the minimal $k$-approximate XWEF allocation achieved in the above counterexample:
\begin{gather*}
	\min_{\alpha \in (0,1)} \text{ } \{\frac{1-\alpha}{\alpha},\frac{\phi\alpha}{1-\alpha},\frac{\phi(1+\phi))}{\alpha(1-\alpha)}\}
\end{gather*}
We finally obtain the maximal approximate factor by equating the three expressions and solving to get $\alpha = 0.44$. Therefore, our maximal approximate factor is 1.272-XWEF.
\end{proof}

\begin{proof}[Proof of Theorem~\ref{thm:xwef_n3}] For the given problem instance:
\begin{center}
\begin{tabular}{ c | c c c c c}
 	 & $b_1$ & $b_2$ & $b_3$ & $b_4$ & $b_5$ \\ 
 	 \hline
	 $c_1$ & $1$ & 1 & $\phi$ & $\phi$ & 1 \\  
 	 $c_2$ & $1$ & 0 & 0 & $1$ & 1 \\
 	 $c_3$ & $\phi$ & 1 & 0 & $\phi$ & $\phi$
\end{tabular}
\end{center}
with $w = (\frac{1}{15},\frac{6}{15},\frac{8}{15})$ we have that the allocation to the three agents with a smallest approximate factor of 1.214 is
$$\B_1 = \{b_1\}, \B_2 = \{b_5\}, \B_3 = \{b_2,b_3,b_4\}$$
\end{proof}

\subsubsection{1WEF Algorithm} \label{appendix:1wef}
\begin{proof}[Proof of Lemma~\ref{thm:1wef_lem1}]
    Since agent $i$ is picked at iteration $t$, it must be that $i \in \arg\min_{i \in \N} \frac{t_i}{w_i}$, and therefore $\frac{t_i}{w_i} \leq \frac{t_j}{w_j}\rightarrow\frac{t_j}{t_i} \geq \frac{w_j}{w_i}$ as $t_i >0$.
\end{proof}

\begin{proof}[Proof of Lemma~\ref{thm:1wef_lem2}]
Consider any iteration $t$ where agent $i$ has been selected to pick the next chore.  Let the number of times $i$ appears in the picking sequence before $j$'s first pick be denoted as $\tau_0$, and further denote the number of times $i$ picks between $j$'s $k$ and $(k+1)$-th pick by $\tau_k$.  For $r \in [t_j]$, we have that $\sum_{x=1}^{r}\tau_x$ and $r$ are the number of chores selected by $i$ before $j$ selects their $(r+1)$-st item and number of items $j$ has selected respectively, and thus Lemma~\ref{thm:1wef_lem1} gives us
\begin{gather*}
	\sum_{x=1}^{r}\tau_x \leq r\phi \text{ for all } r \in [t_j]
\end{gather*}
where we let $\phi := \frac{w_i}{w_j}$.

Additionally, we note that anytime agent $i$ was selected, they picked one of the items they valued least among the remaining chores, including items picked later by agent $j$. Thus, if we let $\beta_x$ denote the cost of the chore picked by agent $j$ in their $x$-th pick and $\alpha_y^x$ be the cost of a chore $i$ picks between $j$'s $(x-1)$ and $x$-th pick, then for all $y \in \tau_x$:
\begin{align*}
	\alpha_{y}^{x} &\leq \text{min} \{\beta_x,\beta_{x+1}, ..., \beta_{t_j}\}
\end{align*}
Subsequently, since $\tau_x \geq 0$ for any $x \in [t_j]$, then
\begin{align}
	\rightarrow \sum_{y=1}^{\tau_x} \alpha_y^x &\leq \tau_x \cdot \text{min} \{\beta_x,\beta_{x+1}, ..., \beta_{t_j}\} \label{ineq.alpha}
\end{align}
We claim that for each $r \in [t_j]$,
\begin{align*}
	\sum_{x=1}^r\sum_{y=1}^{\tau_x}\alpha_y^x \leq &\phi \cdot \left(\beta_{last} + \sum_{x=1}^r\beta_x\right) + \left(\sum_{x=1}^r \tau_x -r\phi\right) \cdot \text{min}\{\beta_r, ..., \beta_{t_j}\} 
\end{align*}
where $\beta_{last}$ is the cost associated with agent $j$'s chore allocated in the final loop of Algorithm~\ref{alg:1wef}. We proceed in proving the claim is true by induction. For the base case of $r = 1$, we utilize inequality (\ref{ineq.alpha}) to get
\begin{align*}
	\sum_{y=1}^{\tau_1} \alpha_y^1 &\leq \tau_1 \cdot \text{min} \{ \beta_1,\beta_{2}, ..., \beta_{t_j} \} \\
	&= \phi \cdot \text{min} \{ \beta_1,\beta_{2}, ..., \beta_{t_j}\} + (\tau_1 - \phi) \cdot \text{min} \{\beta_1,\beta_{2}, ..., \beta_{t_j} \} \\
	&\leq \phi \cdot (\beta_{last} + \beta_1) + (\tau_1 - \phi) \cdot \text{min} \{ \beta_1,\beta_{2}, ..., \beta_{t_j} \}
\end{align*}
where the second from the fact that $$\beta_{last} + \beta_1 \geq \text{min} \{ \beta_1,\beta_{2}, ..., \beta_{t_j}\}.$$ For the inductive step, we assume the claim holds for $r-1$ and we will prove it for $r$:
\begin{align*}
	\sum_{x=1}^r &\sum_{y=1}^{\tau_x}\alpha_y^x = \sum_{y=1}^{\tau_r}\alpha_y^r + \sum_{x=1}^{r-1}\sum_{y=1}^{\tau_x}\alpha_y^x \\
	\leq &\sum_{y=1}^{\tau_r}\alpha_y^r + \phi \cdot \left(\beta_{last} + \sum_{x=1}^{r-1}\beta_x\right) + \left(\sum_{x=1}^{r-1}\tau_x - (r-1)\phi\right) \cdot \text{min} \{\beta_{r-1},\beta_{r}, ..., \beta_{t_j} \} \\
	\leq &\text{ } \tau_r \cdot \text{min} \{\beta_r,\beta_{r+1}, ..., \beta_{t_j} \} + \phi \cdot \left(\beta_{last} + \sum_{x=1}^{r-1}\beta_x\right) + \left(\sum_{x=1}^{r-1}\tau_x - (r-1)\phi\right) \cdot \text{min} \{\beta_{r-1},\beta_{r}, ..., \beta_{t_j} \} \\
	\leq &\text{ } \phi \cdot \left(\beta_{last} + \sum_{x=1}^{r-1}\beta_x\right) + \left(\sum_{x=1}^{r}\tau_x - (r-1)\phi\right) \cdot \text{min}_{r \le \sigma \le t_j} \{\beta_\sigma \} \\
	\leq &\text{ } \phi \cdot \left(\beta_{last} + \sum_{x=1}^{r}\beta_x\right) + \left(\sum_{x=1}^{r}\tau_x - r\phi\right) \cdot \text{min}_{r \le \sigma \le t_j} \{\beta_\sigma \} \}
\end{align*}
Where the first inequality invokes the inductive hypothesis, the second uses (\ref{ineq.alpha}), and the final two steps come from combining and rearranging the terms within the summations as well as the fact that $\text{min} \{\beta_{r-1},\beta_{r}, ..., \beta_{t_j} \}$ is necessarily less than or equal to  $\text{min} \{\beta_r,\beta_{r+1}, ..., \beta_{t_j} \}$.
Lastly, we can see that if we take $r = t_j$, we can once again invoke \ref{ineq.alpha} to obtain
\begin{align*}
	\sum_{x=1}^{t_j}\sum_{y=1}^{\tau_x}\alpha_y^x \leq &\text{ } \phi \cdot \left(\beta_{last} + \sum_{x=1}^{t_j}\beta_x\right) \\
	&\quad + \left(\sum_{x=1}^{t_j} \tau_x - t_j\phi\right) \cdot \text{min}\{\beta_{t_j}, \beta_{t_j+1}, ..., \beta_{t_j}\} \\
	\leq &\text{ } \phi \cdot \left(\beta_{last} + \sum_{x=1}^{t_j}\beta_x\right)
\end{align*}
Observing that $\sum_{x=1}^{t_j}\sum_{y=1}^{\tau_x}\alpha_y^x$ and $\sum_{x=1}^{t_j}\beta_x$ are the partial allocations at iteration $t$ to agents $i$ and $j$ respectively, we can see that on each iteration of the weighted round-robin protocol, agent $i$ is 1WEF up to the final chore allocated in the second loop ($\beta_{last}$) to agent $j$.
\end{proof}

\end{document}